\documentclass{llncs}

\usepackage{amssymb,amsmath,amsfonts}
\usepackage{verbatim}
\usepackage{times,citesort}
\usepackage{graphicx,subfigure}

\newtheorem{pp}{}

\let\doendproof\endproof
\renewcommand\endproof{~\hfill\qed\doendproof}

\newcommand{\charge}{\gamma}

\title{Threshold-Coloring and Unit-Cube Contact Representation of Graphs}

\author{
Md.~Jawaherul~Alam\inst{1}
\and
Steven~Chaplick \inst{2}
\and
Ga\v{s}per~Fijav{\v z}\inst{3}
\and
Michael~Kaufmann\inst{4}
\and
Stephen~G.~Kobourov\inst{1}
\and
Sergey~Pupyrev\inst{1}
}

\institute{
    Department of Computer Science, University of Arizona,
    Tucson, AZ, USA
\and
    Department of Applied Mathematics, Charles University, Prague, Czech Republic
\and
	Faculty of Computer and Information Science,
	University of Ljubljana, Slovenia
\and
    Wilhelm-Schickhard-Institut f\"ur Informatik, Universit\"at T\"ubingen,
    T\"ubingen, Germany
}

\begin{document}

\maketitle

\begin{abstract}

In this paper we study {\em threshold coloring} of graphs, where the vertex colors represented by integers are used to describe any spanning subgraph of the given graph as follows. Pairs of vertices with near colors imply the edge
between them is present and pairs of vertices with far colors imply the edge is absent.
Not all planar graphs are threshold-colorable, but several subclasses, such as trees, some planar grids, and planar graphs
 without short cycles can always be threshold-colored.
Using these results we obtain unit-cube contact representation of several subclasses of planar graphs.
Variants of the threshold coloring problem are related to well-known graph coloring and other graph-theoretic problems. Using these relations we show the NP-completeness for two of these variants, and describe a polynomial-time algorithm for another.
\end{abstract}

\section{Introduction}

Graph coloring is among the fundamental problems in graph theory.
Typical applications of the problem and its generalizations are in job scheduling,
 channel assignments in wireless networks, register allocation in compiler
optimization and many others~\cite{Roberts91}.
In this paper
we consider a new graph coloring problem in which we assign colors (integers) to the vertices of a graph $G$ in order to define a spanning subgraph $H$ of $G$.
In particular, we color the vertices of $G$ so that for each edge of $H$, the
two endpoints are near, i.e., their distance is within a given ``threshold'', and for each edge
of $G\setminus H$, the endpoints are far, i.e., their distance greater than the threshold;
see Fig~\ref{fig:threshold-cube}.

The motivation of the problem is twofold. First, such coloring can be used for the Frequency Assignment Problem~\cite{hale80},
which asks for assigning frequencies to transmitters in radio
networks so that only specified pairs of transmitters can communicate with each other.
Second, such coloring can be used in the context of the geometric problem of
\textit{unit-cube contact representation} of planar graphs.
In such a representation of a graph, each vertex
is represented by a unit-size cube and each edge is realized
by a common boundary with non-zero area between the two corresponding cubes. Finding classes
of planar graphs with unit-cube contact representation was recently posed as an open question
by Bremner~{\em et~al.}~\cite{Bremner12}.
In this paper we partially address this problem as
an application of our coloring problem in the following way.
Suppose a planar graph $G$
has a unit-cube contact representation where one face
of each cube is co-planar; see Fig.~\ref{fig:threshold-cube}(a). Assume that we can
define a spanning subgraph $H$ of $G$ by our particular vertex coloring. We show that
it is possible to compute a unit-cube contact representation of $H$ by lifting the cube for each vertex
$v$ by the amount equal to the color of $v$ (where the size or side-length of the cubes are
roughly equal to the threshold); see Fig.~\ref{fig:threshold-cube}(b).

\begin{figure}[t]
\centering
\includegraphics[width=\textwidth]{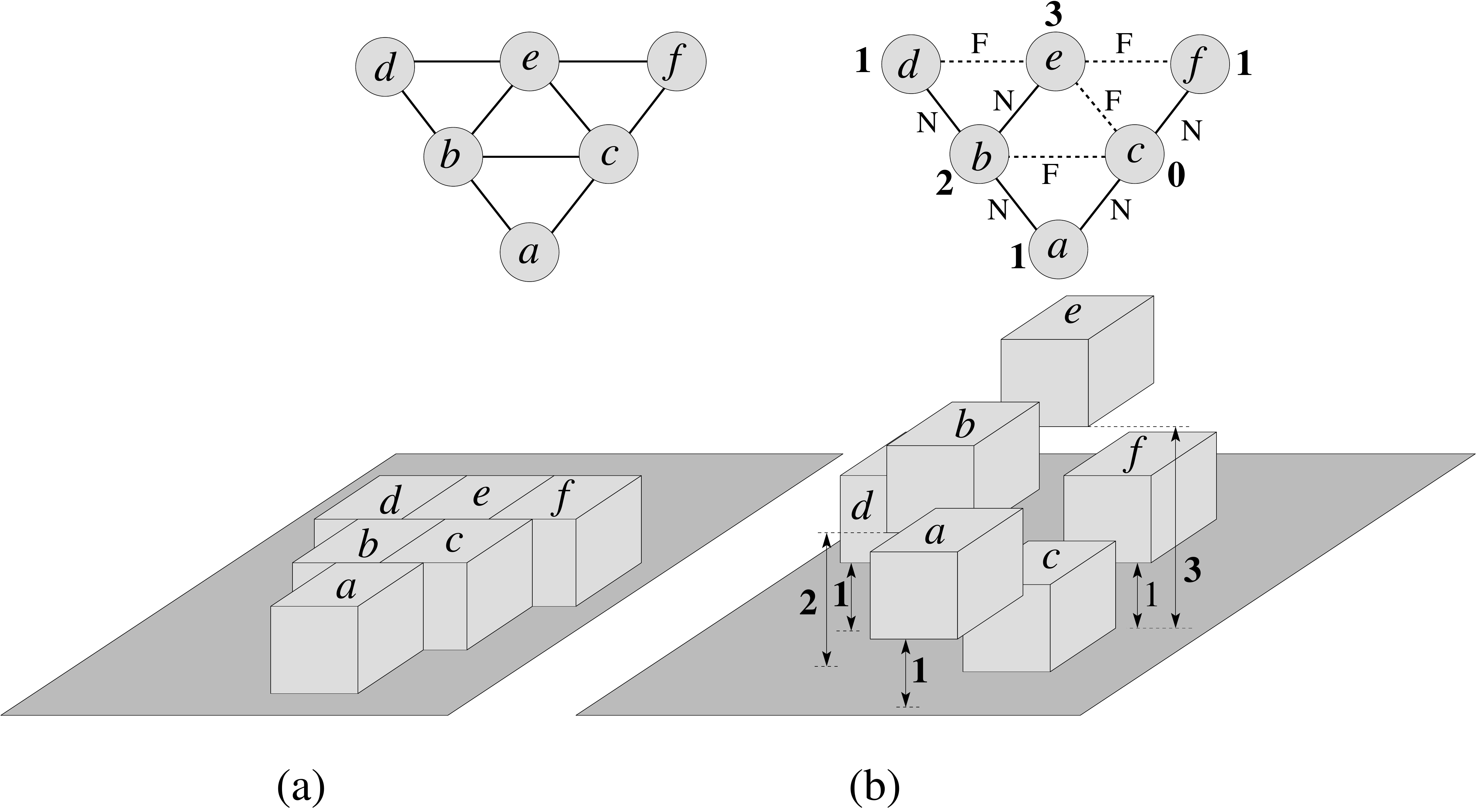}
\caption{(a) A planar graph $G$ and its unit-cube contact representation where the bottom faces of all cubes are
co-planar, (b) a spanning subgraph $H$ of $G$ with
a $(4,1)$-threshold-coloring and its unit-cube contact representation.
 Far edges are shown dashed, near edges are shown solid.}
\label{fig:threshold-cube}
\end{figure}

\subsection{Problem Definition}

An \emph{edge-labeling} of graph $G=(V,E)$ is a mapping $\ell: E \rightarrow \{N,F\}$
assigning labels $N$ or $F$ to each edge of the graph; we informally name edges labeled with $N$ as the
\emph{near} edges, and edges labeled with $F$ as the
\emph{far} edges. Note that such an edge labeling of $G$ defines a partition of the edges
$E$ into near and far edges. By abuse of notation the pair $\{N,F\}$  also denotes this partition.

Let $r\ge 1$ and $t\ge 0$ be two integers and let $[1\dots r]$ denote a set of $r$ consecutive integers.
For a graph $G=(V,E)$ and an edge-labeling $\ell: E \rightarrow \{N,F\}$ of $G$, a
\textit{$(r,t)$-threshold-coloring} of $G$ with respect to $\ell$ is a \textit{coloring}
$c:V\rightarrow [1\dots r]$ such that for each edge
$e=(u,v)\in E$, $e\in N$ if and only if $|c(u)-c(v)|\le t$.
We call $r$ and $t$ the \emph{range} and the \emph{threshold}.
Note that the set of near edges defines a spanning subgraph $H=(V,N)$ of $G$, where $H$ is a \textit{spanning subgraph} of graph $G$ if it contains all vertices of $G$.
 $H$ is a \textit{threshold subgraph} of $G$ if there exists such a threshold-coloring.

A graph $G$ is \emph{total-threshold-colorable} if for every edge-labeling $\ell$ of $G$
there exists an $(r,t)$-threshold-coloring of $G$ with respect to $\ell$ for some $r\ge 1, t\ge 0$.
Informally speaking, for every partition of edges of $G$ into near and far edges, we can
produce vertex colors so that endpoints of near edges receive near colors, and
endpoints of far edges receive colors that are far apart.
A graph $G$ is \emph{$(r,t)$-total-threshold-colorable} if it is total-threshold-colorable
for the range $r$ and threshold $t$. In this paper we focus on the following problem variants:

\def\Q#1{\ifnum#1=1
            \textbf{Total-Threshold-Coloring Problem}\else
           \ifnum#1=2
             \textbf{Threshold-Coloring Problem}\else
             \ifnum#1=3
               \textbf{Fixed-Threshold-Coloring Problem}\else
               \ifnum#1=4
                 \textbf{Exact-Threshold-Coloring Problem}\fi\fi\fi\fi}

\begin{problem}
(\Q1) Given a graph $G$, is $G$ total-threshold-colorable,
that is, is every spanning subgraph of $G$ a threshold subgraph of $G$?
\end{problem}

The problem is closely related to the question about whether a particular spanning graph $H$ of $G$ is threshold-colorable.

\begin{problem}
(\Q2) Given a graph $G$ and a spanning subgraph $H$,
is $H$ a threshold subgraph of $G$ for some integers $r\ge 1, t\ge 0$?
\end{problem}

Another interesting variant of the threshold-coloring is the one in which we specify that the graph $G$
is the complete graph. In this case we call $H$ an \textit{exact-threshold} graph if $H$ is a threshold
subgraph of the complete graph $G$ for some integers $r \ge 1, t\ge 0$.

\begin{problem}
(\Q4) Given a graph $H$, is $H$ an exact-threshold graph?
\end{problem}

In the final variant of the problem we assume that the threshold and the
range are the part of the input:

\begin{problem}
(\Q3) Given a graph $G$, a spanning subgraph $H$, and
integers $r\ge 1, t \ge 0$, is $H$ $(r,t)$-threshold-colorable?
\end{problem}

\subsection{Related Work}

Many problems in graph theory
deal with coloring or assigning labels to the vertices of a graph; many graph
classes are defined based on such coloring and labeling; see~\cite{Brandstadt99} for an excellent survey.
To the best of our knowledge, total-threshold-colorability defines a new
class of graphs. Here we mention two closely related classes: threshold graphs and difference graphs.
\emph{Threshold graphs} are
ones for which there is a real number $S$ and for every vertex $v$ there is a real weight
$a_v$ such that $(v,w)$ is an edge if and only if $a_v + a_w \ge S$~\cite{Mahadev95}.
A graph is a \textit{difference graph} if there is a real number
$S$ and for every vertex $v$ there is a real weight $a_v$ such that $|a_v| < S$ and
$(v,w)$ is an edge if and only if $|a_v - a_w| \ge S$~\cite{Hammer90}. Note that for both
classes the threshold (real number $S$) defines edges between all pairs of
vertices, while in our setting the threshold defines only the edges of a
graph $G$, which is not necessarily a complete graph.
Both threshold and difference graphs
can be characterized in terms of forbidden induced subgraphs. For
our problem such a characterization is unknown. For details on threshold and difference graphs,
see~\cite{Mahadev95}.

Another related graph coloring problem is the \emph{distance constrained graph
labeling}. Here the goal is to find $L(p_1,\dots, p_k)$-labeling of the vertices of a graph
so that for every pair of vertices at distance at most $i\le k$ we have that the difference of their
labels is at least $p_i$. The most studied variant is $L(2,1)$-labeling~\cite{Griggs92,Fiala05}.
In~\cite{Griggs92} it was shown that minimizing the number of labels in $L(2,1)$-labeling is
NP-complete, even for graphs with diameter 2. Further, it shown that it is also NP-complete
to determine whether a labeling exists with at most $k$ labels for every fixed integer $k\ge 4$~\cite{Fiala01}.

A threshold-coloring of a planar graph can be used to find a contact representation of the graph with cuboids
(axis aligned boxes) in 3D.
Thomassen~\cite{Thomassen86} shows that any planar
graph has a proper contact representation by cuboids in 3D. In a \textit{contact representation}
 of a graph, the vertices are represented by cuboids (or other polygonal shapes) and the edges
 are realized by a common boundary of the two corresponding cuboids. A contact representation
 is \textit{proper} if for each edge the corresponding common boundary has non-zero area.
 Felsner and Francis~\cite{Felsner11} prove that any planar graph has a
(non-proper) contact representation by cubes.

Bremner~{\em et~al.}~\cite{Bremner12} proves that the same result does not hold when using only unit cubes.
Our results on threshold-coloring of planar graphs translates to results on classes of planar graphs that can be represented by contact of unit cubes.

\subsection{Our Contribution}

\begin{enumerate}

\item
We study the \Q1 for various subclasses of planar graphs. In particular, we show that
 several subclasses of planar graphs are threshold-colorable (e.g., trees, hexagonal grids,
 planar graphs without any cycles of length $\le 9$) and several subclasses are not
 (e.g., triangular grid, 4-3 grid). Our results are summarized in Table~\ref{tbl:results}.

\def\tableimage#1{\includegraphics[width=1cm]{#1}}

\begin{table}[t]
    \centering
    \begin{tabular}{|c|c|c|c|c|c|c|c|c|c|c|l|}
        \hline
	\parbox[c]{1cm}{\vspace{0.8cm}graph classes} &
        \parbox[c]{0.7cm}{Cycle} &
        \parbox[c]{0.7cm}{Tree} &
        \parbox[c]{0.7cm}{Fan} &
        \parbox[c]{1.3cm}{Triangular Grid} &
        \parbox[c]{1.1cm}{Square Grid} &
        \parbox[c]{1.3cm}{Hexagonal Grid} &
        \parbox[c]{1.3cm}{Octagonal-Square Grid} &
        \parbox[c]{1.2cm}{Square-Triangle Grid} &
        \parbox[c]{1.7cm}{Planar Graph w/o Cycles of size $\le 9$} \\

        \cline{2-10}

	&
        \tableimage{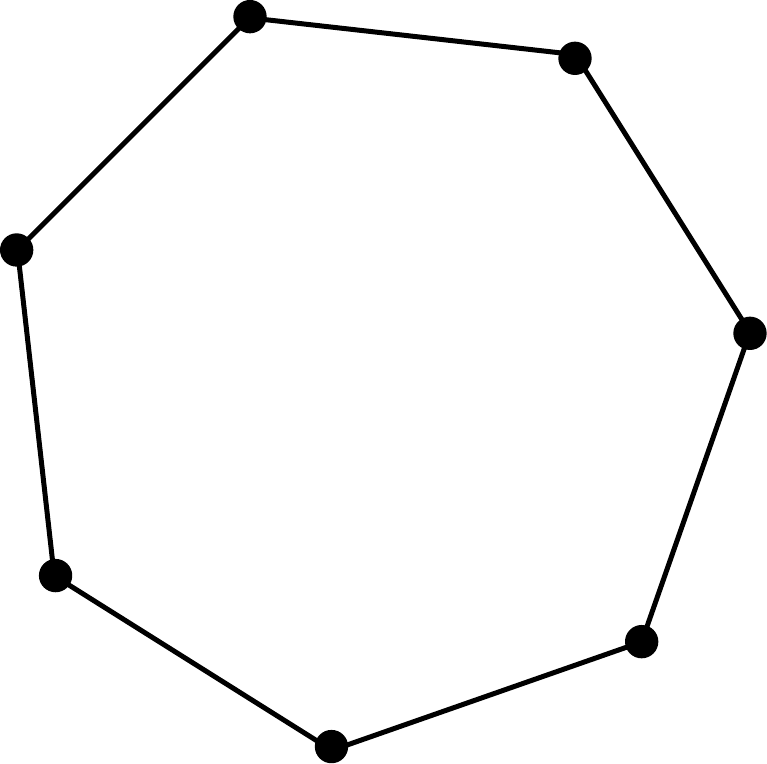} &
        \tableimage{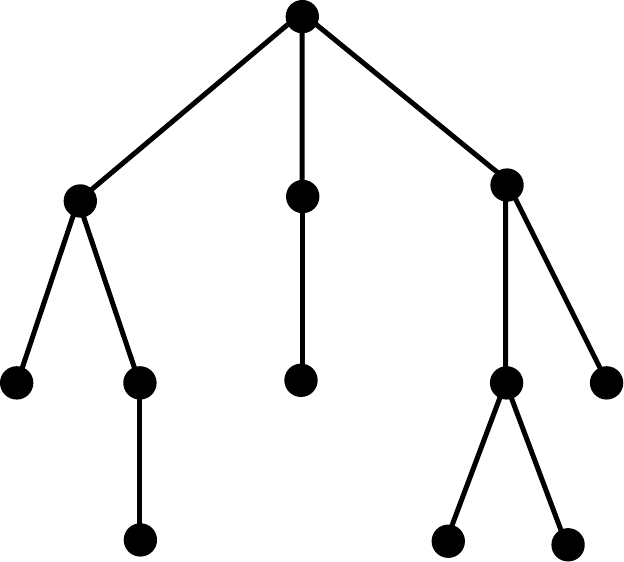} &
        \tableimage{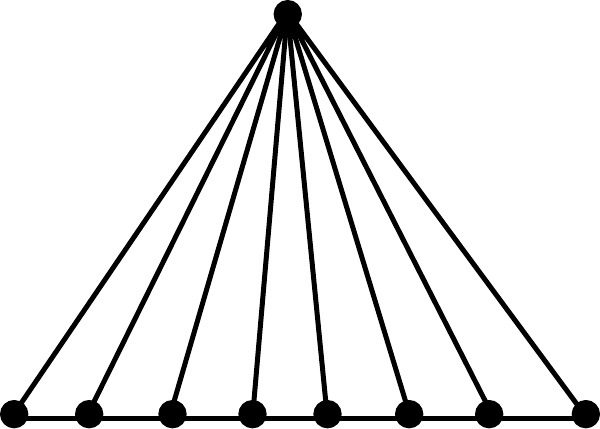} &
        \tableimage{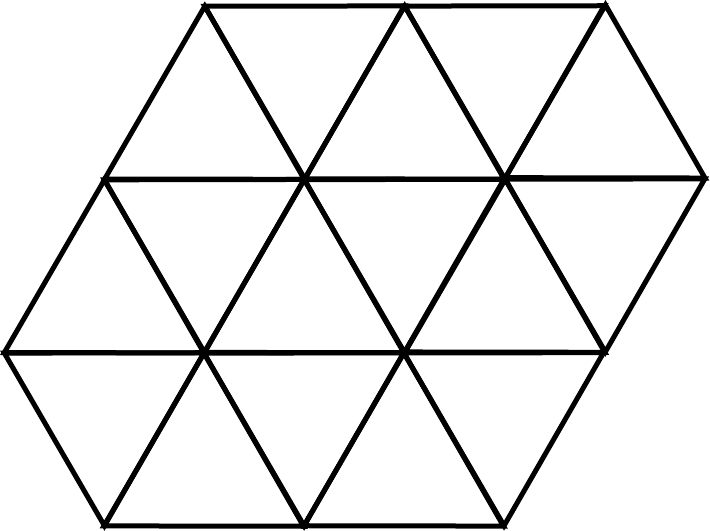} &
        \tableimage{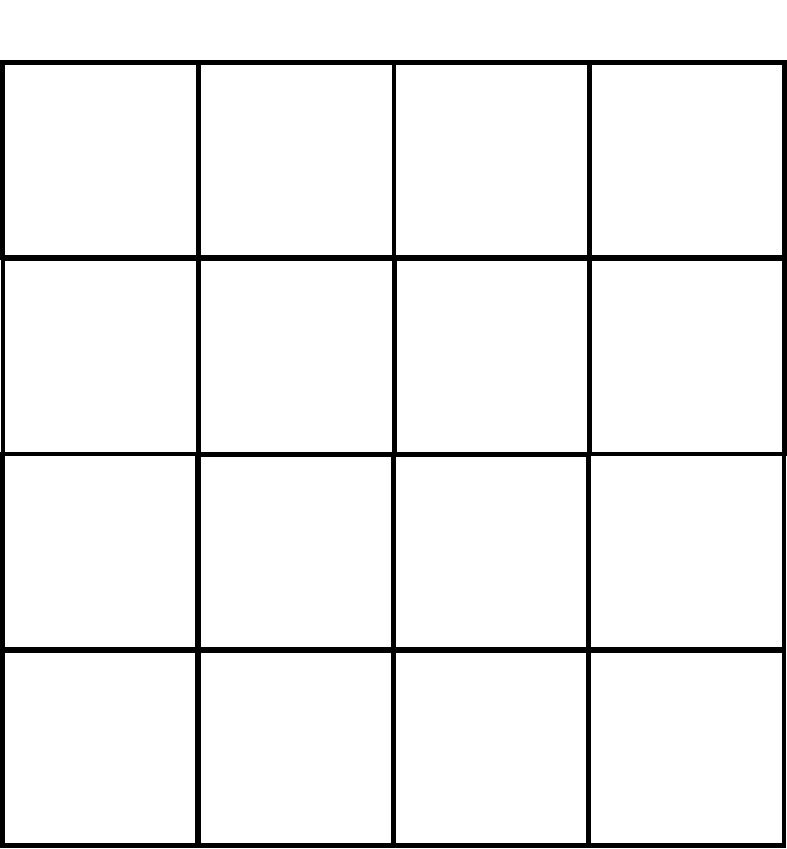} &
        \tableimage{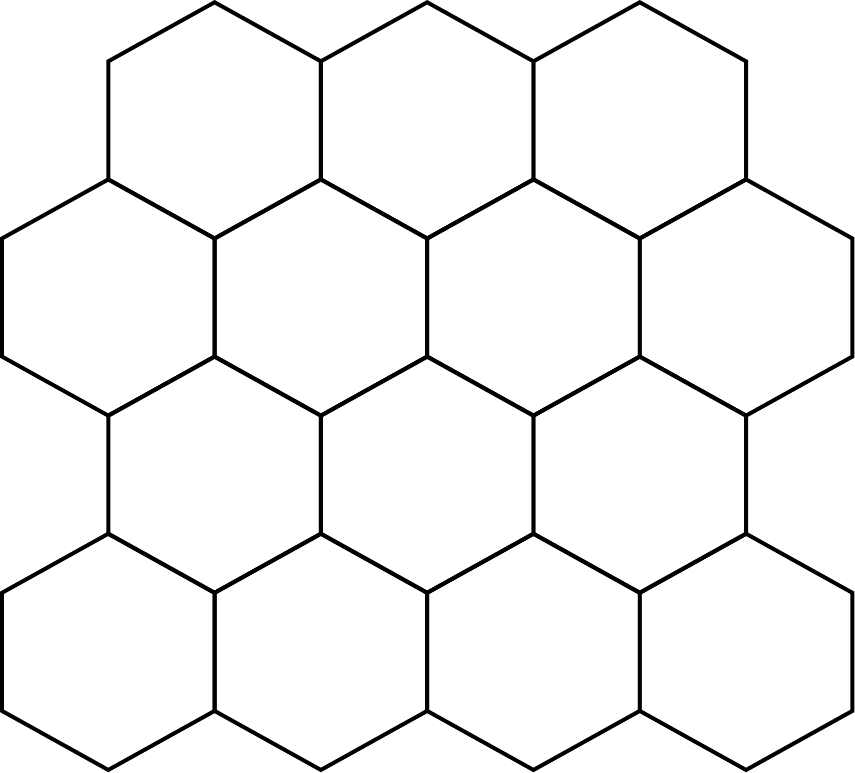} &
        \tableimage{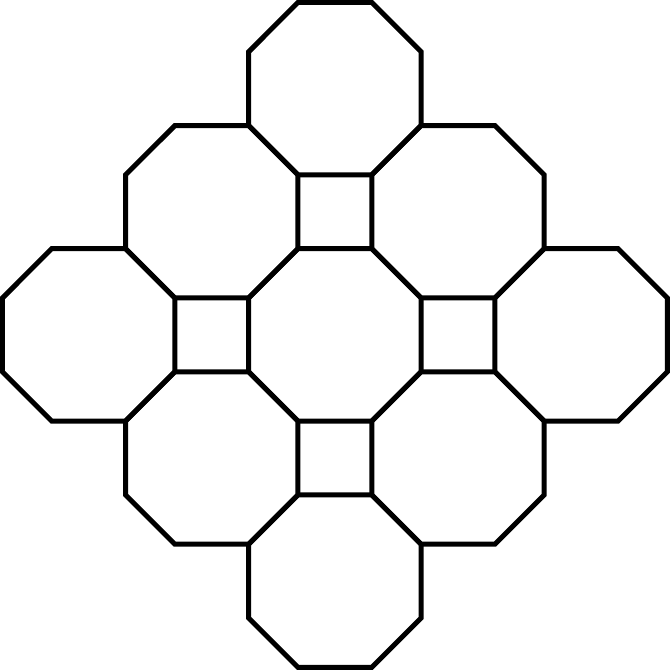} &
        \tableimage{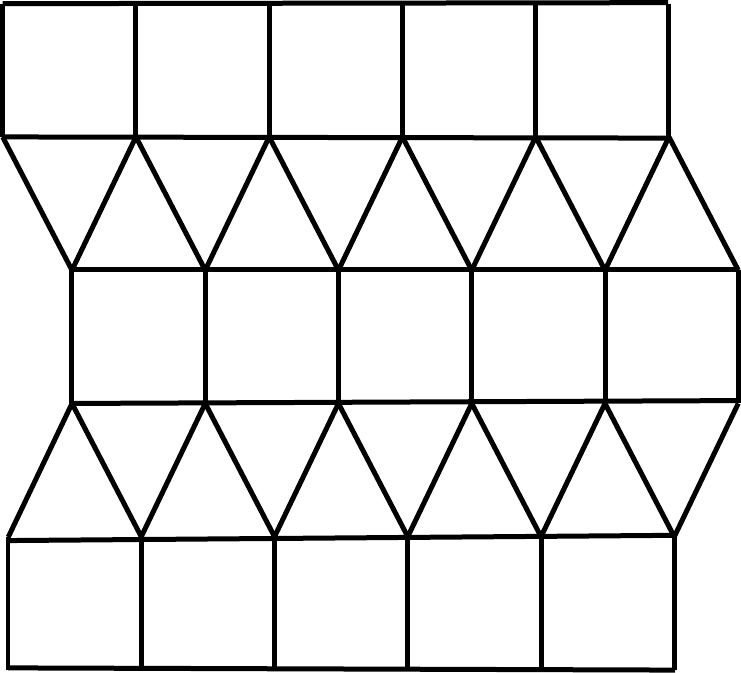} &
        \tableimage{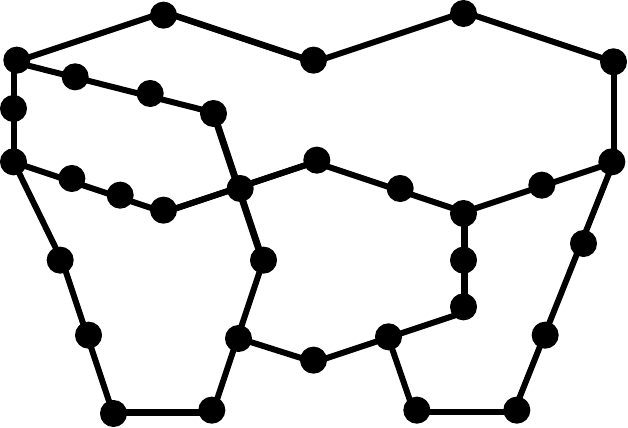} \\
        \hline
	\hline

	\parbox[c]{1.2cm}{threshold coloring} &
	\parbox[c]{0.9cm}{$r=5$, $t=1$} &
	\parbox[c]{0.9cm}{$r=2$, $t=0$} &
	\parbox[c]{0.9cm}{$r=5$, $t=1$} &
        No &
        Open &
	\parbox[c]{0.9cm}{$r=5$, $t=1$} &
	\parbox[c]{0.9cm}{$r=5$, $t=1$} &
        No &
	\parbox[c]{0.9cm}{$r=8$, $t=2$} \\
        \hline

	\parbox[c]{1.2cm}{unit-cube contact} &
        Yes &
        No &
        No &
        Open &
        Yes &
        Yes &
        Yes &
        Open &
        No \\
        \hline
    \end{tabular}
\vspace{.2cm}
    \caption{Results on the \Q1. ``No'' entries in the last row follow from the fact that graphs with vertices of high degrees cannot have unit-cube representation~\cite{Bremner12}.}
    \label{tbl:results}
\end{table}

\item
As an application of the threshold-coloring problem, we address the problem of contact
 representation of planar graphs with unit cubes.
 Given a planar graph, we investigate whether each of its subgraphs has a contact
 representation with unit cubes. We show how we can use the threshold-coloring
 for computing unit-cube contact representations for some subclasses of planar graphs.
 For some other subclasses, we gave algorithms to directly compute unit-cube contact
 representation without using threshold-coloring. Thus we answer some of the open problems
 from~\cite{Bremner12} for some subclasses of planar graphs.
 The last column of Table~\ref{tbl:results} summarizes these results.

\item
Finally we study the relation of the various threshold-coloring problems with other graph-theoretic problems.
Specifically, we show that the \Q2 and the \Q3 are NP-complete
 by reductions from  a graph sandwich problem and the classical vertex coloring problem, respectively. We also
 show that the \Q4 can be solved in linear time since it is equivalent to the
 proper interval graph recognition problem.

 \end{enumerate}

\section{Threshold-Coloring and Other Graph Problems}
\label{sect:relation}

We begin by showing the connections between threshold-colorability and some classical graph-theoretical and graph
coloring problems.

\subsection{Vertex Coloring Problem}

Let $G=(V,E)$ be a graph. We call $G$ $k$-vertex-colorable if there exists a coloring
$c:V\rightarrow [1\dots k]$ such that for any edge $(u,v)\in E$, $c(u)\neq c(v)$, that is, $u$ and $v$
have different colors. Given an input graph $G$ and an integer $k>0$, the vertex coloring
problem asks whether there exists a $k$-vertex-coloring of $G$.

\begin{lemma}
\label{lem:color} Let $G=(V,E)$ be a graph and let $k$ be a positive integer. Define an edge-labeling
 $\ell:E\rightarrow \{N,F\}$ that assigns each edge the label $F$, that is, for each edge $e\in E$,
 $\ell(e)=F$. Then $G$ has a $k$-vertex-coloring if and only if there exists a $(k,0)$-threshold-coloring
 of $G$ with respect to $\ell$.
\end{lemma}

\begin{proof} Let $c:V\rightarrow [1\dots k]$ define a mapping of the vertices of $G$ to the colors
 $[1\dots k]$. Then $c$ is a $k$-vertex-coloring of $G$ $\Leftrightarrow$ for each edge $e=(u,v)\in E$,
 $c(u)\neq c(v)$ $\Leftrightarrow$ for each edge $e=(u,v)\in E$, $|c(u)-c(v)|>0$
 $\Leftrightarrow$ $c$ is a $(k,0)$-threshold-coloring of $G$ with respect to $\ell$.
\end{proof}

\begin{corollary} The \Q3 is NP-complete.
\end{corollary}

\subsection{Proper Interval Representation Problem}

An \textit{interval representation}~\cite{Brandstadt99} for a graph $G=(V,E)$ is one where
each vertex $v$ of $G$ is represented by an interval $I(v)$ of $\mathbb{R}$ such that for any edge $(u,v)\in E$,
the intervals $I(u)$ and $I(v)$ have a non-empty intersection, that is, $I(u)\cap I(v)\neq\emptyset$.
A \textit{proper interval representation}~\cite{Brandstadt99} for $G$ is an interval
representation of $G$
where no interval properly contains another. A \textit{proper interval graph} is one that
has a proper interval representation. Equivalently, a \textit{proper interval graph} is
one that has an interval representation with unit intervals. The problem of proper interval
representation for a graph $G$ asks whether $G$ has a proper interval representation.
The problem has been studied extensively~\cite{FMM95,HSS99}, and it still attracts
attention~\cite{Le12}.

\begin{lemma}
\label{lem:exact-proper} A graph is an exact-threshold graph if and only if it is a proper interval
 graph.
\end{lemma}
\begin{proof} Let graph $H=(V,E)$ be an exact-threshold graph. This implies that
 there are integers $r \ge 1, t\ge 0$ and a mapping $c:V\rightarrow [1\dots r]$ such
 that for any pair $u,v\in V$, $(u,v)\in E \Leftrightarrow |c(u)-c(v)|\le t$
 $\Leftrightarrow |c(u)-c(v)| < t+\epsilon$ with $0<\epsilon<1$ since
 $c(u)$ and $c(v)$ are integers.
 We can find an interval representation of $H$ with unit intervals as follows.
 Choose an arbitrary $\epsilon$ such that $0<\epsilon<1$.
 Define for each vertex $v$ of $H$ an interval $I(v)$ of unit length where the
 left-end has $x$-coordinate $c(v)/(t+\epsilon)$. Then for any two vertices $u$
 and $v$ of $H$, $I(u)$ and $I(v)$ has a non-empty intersection if and only if
 $|\frac{c(u)}{t+\epsilon}-\frac{c(v)}{t+\epsilon}|\le 1$
 $\Leftrightarrow |c(u)-c(v)|\le (t+\epsilon) \Leftrightarrow |c(u)-c(v)| \le t$ since
 $c(u)$ and $c(v)$ are integers. Then $l(u)$ and $l(v)$ has non-empty intersection
 if and only if $(u,v)$ is an edge of $H$.
 Thus these intervals yield an interval representation of $H$.

Conversely, if $H$ has an interval representation $\Gamma$ with unit intervals, we can find an
 exact $(r,t)$-threshold-coloring of $H$ for some integers $r \ge 1, t\ge 0$. Scale
 $\Gamma$ by a sufficiently large factor $t$ such that each end-point of some interval in
 $\Gamma$ has a positive integer $x$-coordinate (after possible translation in the positive
 $x$ direction). Let $r$ be the $x$-coordinate of the right end-point of the rightmost interval
 in this scaled representation. Define a coloring $c:V\rightarrow [1\dots r]$ where for each vertex
 $v$ of $H$, $c(v)$ equals the $x$-coordinate of the left end-point of the interval for $v$.
 Also define the threshold as the scaling factor $t$. It is easy to verify that $c$
 is indeed an $(r,t)$-threshold-coloring.
\end{proof}

Since recognition of proper interval graphs can be done
 in linear time~\cite{HSS99,FMM95,Le12}, we have

\begin{corollary} The \Q4 can be solved in linear time.
\end{corollary}

\subsection{Graph Sandwich Problem}

The graph sandwich problem is defined in~\cite{Golumbic95} as follows.

\begin{problem}
\label{prob:sand} Given two graphs $G_1=(V,E_1)$ and $G_2=(V,E_2)$ on the same vertex
 set $V$, where $E_2\subseteq E_1$, and a property $\Pi$, does there exist a graph $H=(V,E)$
 on the same vertex set such that $E_2\subseteq E\subseteq E_1$ and $H$ satisfies
 property $\Pi$?
\end{problem}

Here $E_1$ and $E_2$ can be thought of as \textit{universal} and \textit{mandatory} sets of
edges, with $E$ sandwiched between the two sets. We are interested
in a particular property for the graph sandwich problem: ``proper interval representability''.
A graph satisfies \textit{proper interval representability} if it admits a proper interval
representation.

\begin{lemma}
\label{lem:threshold-sandwich} Let $G=(V,E_G)$ and $H=(V,E_H)$ be two graphs on the same
 vertex set $V$ such that $E_H\subseteq E_G$. Then the threshold-coloring problem for $G$
 with respect to the edge partition $\{E_H, E_G-E_H\}$ is equivalent to the graph sandwich
 problem for the vertex set $V$, mandatory edge set $E_H$, universal edge set
 $E_H\cup (V\times V - E_G)$ and proper interval representability property.
\end{lemma}
\begin{proof}
Let $E_U$ denote the universal edge set $E_H\cup (V\times V - E_G)$ for the
 graph sandwich problem. Suppose there exists a graph $H^*=(V,E^*)$ such that
 $E_H\subseteq E^*\subseteq E_U$ and $H^*$ has a proper interval representation. Then by
 Lemma~\ref{lem:exact-proper}, there exist two integers $r\ge 0$ and $t\ge 0$ and a
 coloring $c:V\rightarrow [1\dots r]$ such that for any pair $u, v\in V$, $|c(u)-c(v)|\le t$ if
 and only if $(u,v)\in E^*$. We now show that $c$ is in fact a desired
 threshold-coloring for $G$. Consider an edge $e=(u,v)\in E_G$. If $e\in E_H$ then $e\in E^*$
 since $E_H\subseteq E^*$ and hence $|c(u)-c(v)|\le t$. On the other hand if $e\in (E_G-E_H)$,
 $e\notin E_U=E_H\cup (V\times V - E_G)$ and therefore $e\notin E^*$ since
 $E^*\subseteq E_U$. Hence $|c(u)-c(v)|> t$.

Conversely, if there exists integers $r\ge 1$ and $t\ge 0$ such that there is an
 $(r,t)$-threshold-coloring $c:V\rightarrow [1\dots r]$ of $G$ with respect to the edge partition
 $\{E_H,E_G-E_H\}$, then define
 an edge set $E^*$ as follows. For any pair $u,v\in V$, $(u,v)\in E^*$ if and only if
 $|c(u)-c(v)|\le t$. Clearly the graph $H^*=(V,E^*)$ has an exact $(r,t)$-threshold-coloring
 and hence by Lemma~\ref{lem:exact-proper}, $H^*$ has a proper interval representation.
 Furthermore for any edge $e=(u,v)\in E_H$, $|c(u)-c(v)|\le t$ and hence $e\in E^*$. Thus
 $E_H\subseteq E^*$. Again if $e\in E^*$ then $|c(u)-c(v)|le t$. Therefore either $e\in E_H$
 or $e\notin E_G\Rightarrow e\in (V\times V-E_G)$.
 Hence $e\in (E_H\cup (V\times V - E_G))= E_U$. Thus $E^*\subseteq E_U$. Therefore $E^*$
 is sandwiched between the mandatory and the universal set of edges and $H^*$ has a proper
 interval representation.
\end{proof}

Golumbic~et~al.~\cite{GKS94} proved that the graph sandwich problem is
NP-complete for the proper-interval-representability property. Hence
we have

\begin{corollary} The \Q2 is NP-complete.
\end{corollary}

Summarizing the results in this section, we have the following theorem.

\begin{theorem}
\label{th:labeling} The \Q2 and the \Q3 are NP-complete while the \Q4 can be solved in
 linear time.
\end{theorem}

\section{Total-Threshold-Coloring of Graphs}
\label{sect:threshold}

In this section we address the \Q1: is a given graph $G$ total-threshold-colorable, that is, can every spanning subgraph
of $G$ be represented by appropriately coloring the vertices of $G$?

First note that not every graph (not even every planar graph) is total-threshold-colorable. Suppose that $G=K_4$, and
we would like to represent a subgraph where four of the edges remain and span a $4$-cycle, while the other two edges
are removed (edge-partitioning $\{N,F\}$). Assume that there exists an $(r,t)$-threshold-coloring with colors $c_1, c_2, c_3, c_4$ for vertices $v_1,v_2,v_3,v_4$ respectively.
Without loss of generality assume $c_4$ is the highest color and $(v_1,v_4) \in F$, hence
 also $(v_2,v_3) \in F$. Also assume $c_3\ge c_2$ and consequently
$c_4-c_2 \ge c_3 - c_2$. The left side of the inequality should be at most $t$,
and the right side strictly greater than $t$, which cannot be accomplished by any choice of the range and the threshold.

Next we investigate several subclasses of planar graphs. For each of them we either give an algorithm
to find an $(r,t)$-threshold-coloring of any graph in that class with respect to each edge-partition
for some $r \ge 1, t\ge 0$; or we give an example of a graph in that
class and the edge-partition for which there is no threshold-coloring.

\subsection{Paths, Cycles, Trees, Fans}
For \emph{paths} and \emph{trees} there is a trivial coloring with threshold $t=0$ and two colors.
Choose an arbitrary vertex as the root and color it $0$. Color $1$ all vertices with an odd number of far edges on the
shortest path to the root. Color $0$ all vertices with an even number of far edges to the root. Then all vertices connected
by a near edge of $G$ get the same color, and vertices connected by a far edge
get different colors; see Fig.~\ref{fig:elem}(a).

\begin{figure}[t]
    \center
    \includegraphics[height=2.5cm]{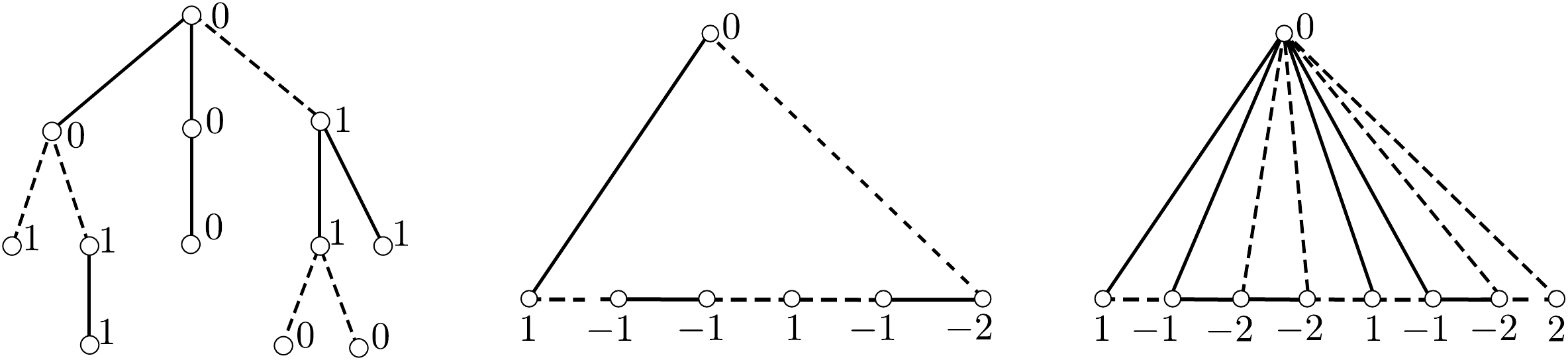}
    \caption{Threshold-coloring of trees, cycles, and fans.}
    \label{fig:elem}
\end{figure}


For \emph{cycles} and \emph{fans} there is a coloring scheme with threshold $t=1$ and five colors.
A \emph{fan} is obtained from a path $P$ by adding a new vertex $v$ connected to all vertices of the path.
We use colors $\{-2, -1, 0, 1, 2\}$ to color a fan.
The vertices of $P$ are colored by $-1$ and $1$, and $v$ is colored by $0$. After this initial coloring
some of the far edges $(u, v), u\in P$ might have $|c(u) - c(v)| = 1$. We fix it by changing the color
of $u$ from $1$ to $2$ or from $-1$ to $-2$; see Fig.~\ref{fig:elem}(c).
It is easy to see that the same algorithm can be applied to color a cycle; see Fig.~\ref{fig:elem}(b).

\subsection{Triangular Grid}
In a triangular grid (planar weak dual of a hexagonal grid) all faces are triangles and internal vertices
have degree 6. It is easy to show that a triangular grid is not total-threshold-colorable. Consider the graph
with vertices $v_0,v_1,v_2,u_0,u_1,u_2$, where each vertex
$u_i$ is adjacent to $v_{i+1}$ and $v_{i+2}$ (mod 3); see Fig.~\ref{fig:trian}.
Let $F = \{(v_0, v_1)$, $(v_1, v_2)$, $(v_2, v_0)\}$, and let $N$ contain the
 remaining $6$ edges. Assume that
there exists a $(r,t)$-threshold-coloring $c$. Without loss of generality, let
$c(v_0) < c(v_1) < c(v_2)$. Now on one hand $c(v_2)-c(v_0) > 2t$ and on the other
$c(v_2)-c(v_0) \le |c(v_2)-c(u_1)|+|c(u_1)-c(v_0)| \le 2t$, which is impossible.
This also proves that outerplanar graphs are not total-threshold-colorable in general.

\begin{figure}[t]
    \center
	\subfigure[]{
    \includegraphics[width=2.5cm]{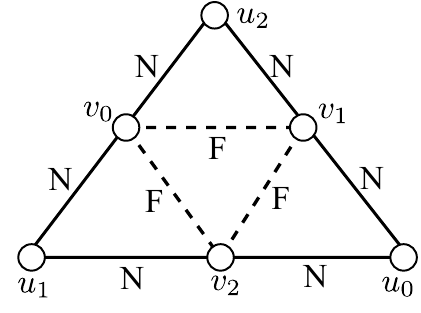}
    \label{fig:trian}}
~~~~~~~~~~
    \subfigure[]{
    \includegraphics[height=2.5cm]{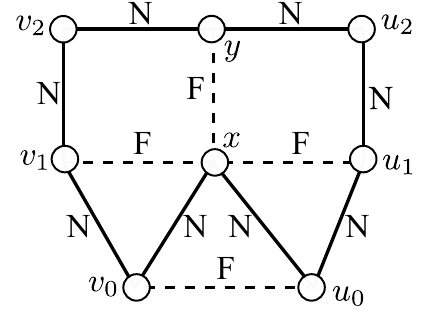}
    \label{fig:43}}
    \caption{Graphs which are not total-threshold-colorable.}
\end{figure}

\subsection{Hexagonal Grid}

In a hexagonal grid all faces are 6-sided and internal vertices have degree 3 (planar weak dual of the triangular grid).
Here we show that the hexagonal grid is total-threshold-colorable with $r=5$ and $t=1$.
We begin with a simple lemma about (5,1) color space.
For convenience, we use colors $\{-2, -1, 0, 1, 2\}$.

\begin{lemma}
\label{pp:p151}
Let $P_2=\{v_0,v_1,v_2\}$ be a path of length 2. Then for any edge-labeling of $P_2$ and a fixed color $k \in \{-2, -1, 1, 2\}$, there is a
threshold-coloring $c$ of $P_2$ with threshold $t=1$, where $c(v_0) = 0$, $c(v_2)=k$ and $c(v_1 ) \in \{-2, -1, 1, 2\}$.
\end{lemma}
\begin{proof}
Depending on whether the edge $(v_0, v_1)$ is near or far, choose $c(v_1)$ to be $1$ or $2$. If the label of $(v_1, v_2)$
disagrees with the colors of $v_1$ and $v_2$ then we change the sign of $c(v_1)$.
\end{proof}

\begin{figure}[b]
    \center
	\subfigure[]{
	\includegraphics[height=2.4cm]{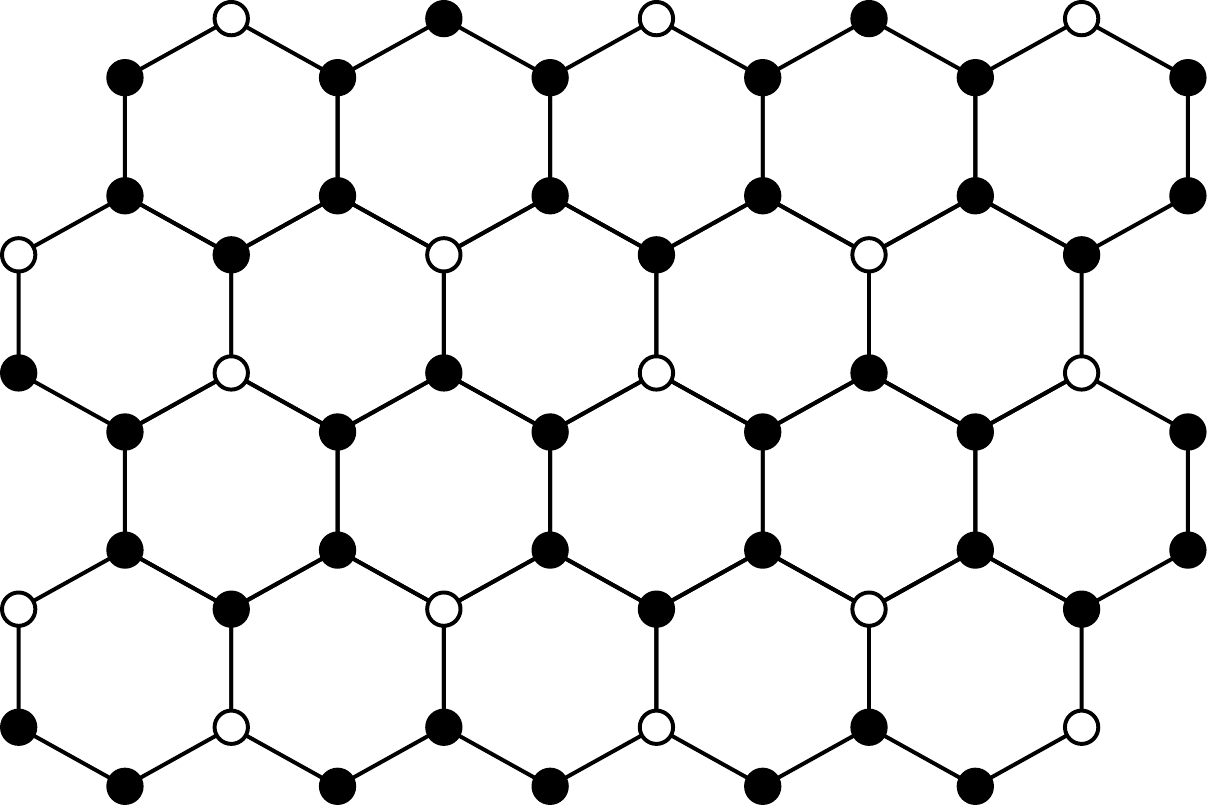}
    \label{fig:hex51a}}
~
    \subfigure[]{
    \includegraphics[height=2.4cm]{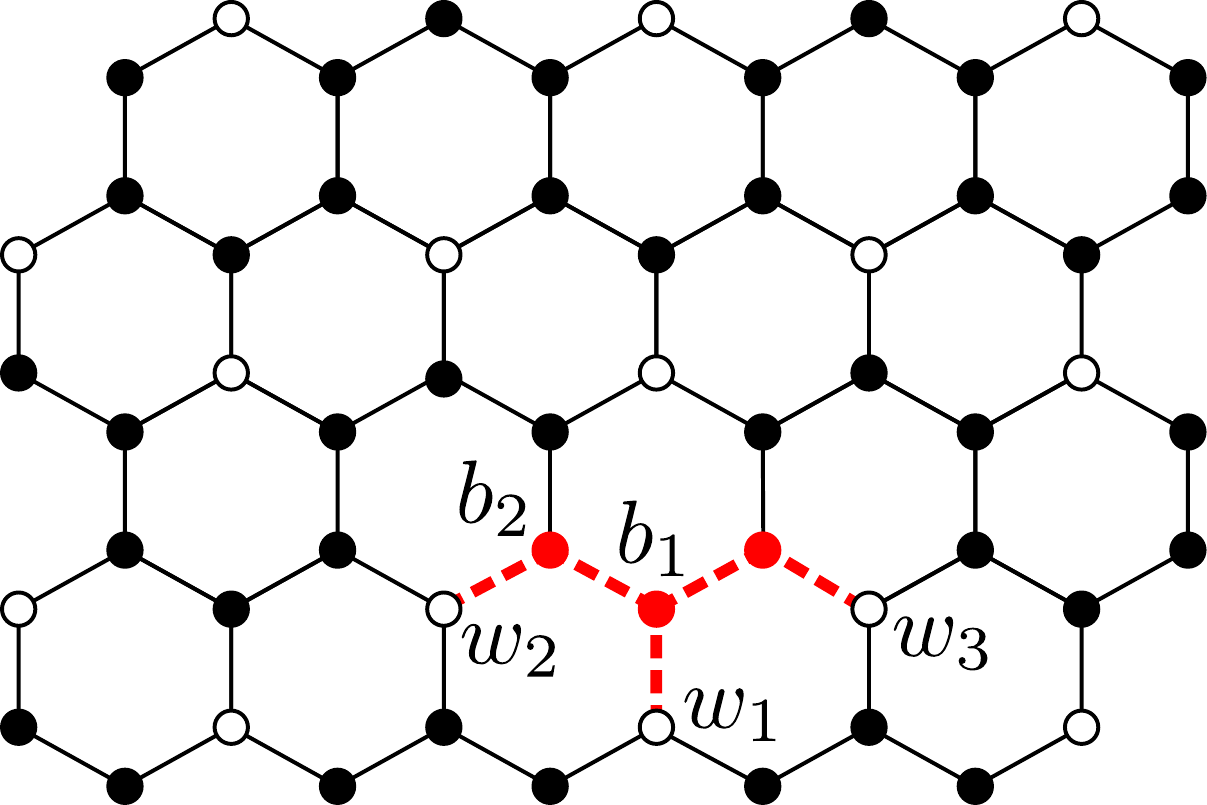}
    \label{fig:hex51b}}
~
    \subfigure[]{
    \includegraphics[height=2.4cm]{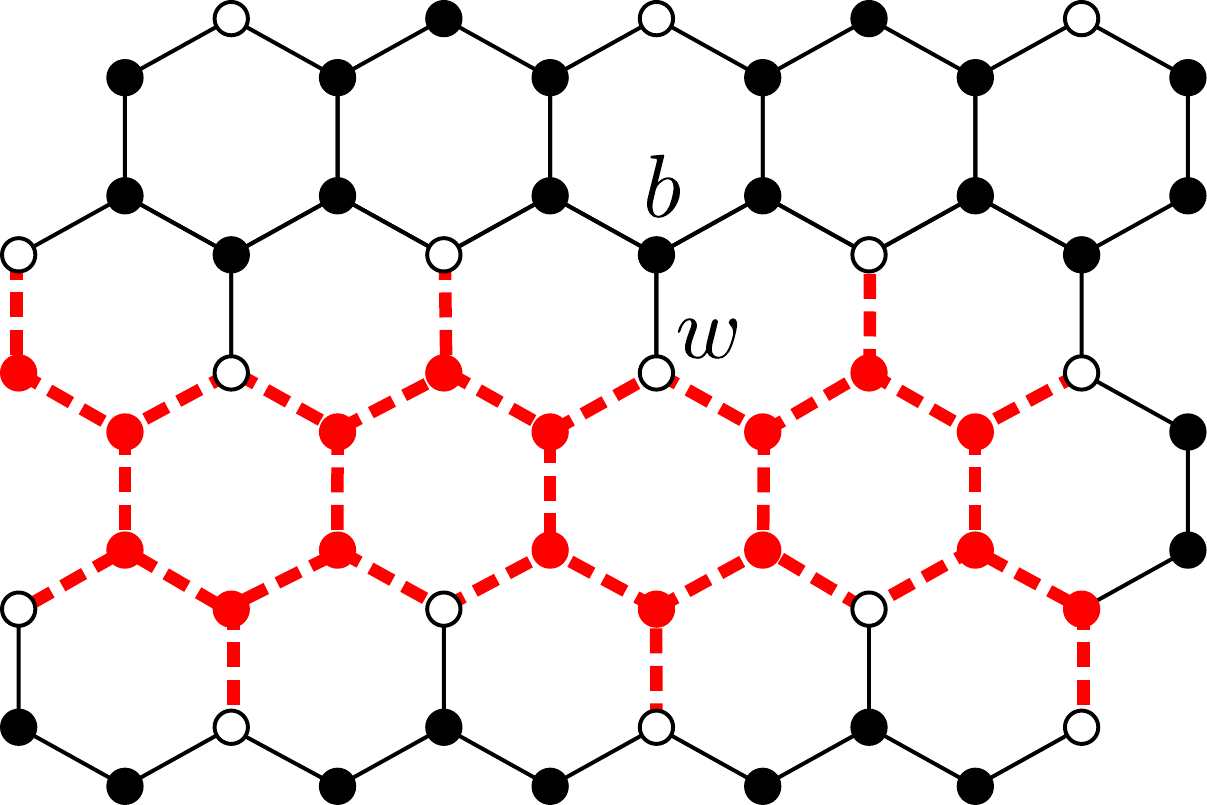}
    \label{fig:hex51c}}
\caption{
Total-threshold-coloring of the hexagonal grid.
(a)~White vertices get color 0, black vertices get one of the colors $-2, -1, 1, 2$.
(b)~A color assignment to $b_1$ can be extended to vertices $w_2$ and $w_3$ based on
the labels of the red dashed edges.
(c)~The process assigns colors for the red vertices.}
\end{figure}

\begin{lemma}
\label{lm:hex51}
Any hexagonal grid is $(5,1)$-total-threshold-colorable.
\end{lemma}

\begin{proof}
The coloring is done in two steps. In the first step we assign color $0$ for a set of independent vertices of $G$ as shown
in Fig.~\ref{fig:hex51a}, where the colored vertices are white. Note that no two white vertices have
a shortest path of length less than $3$.

In the second step we find a coloring of the remaining black vertices, using only four colors $\{-2, -1, 1, 2\}$.
Let $w_1$ be a white vertex. We randomly choose one of its black neighbors $b_1$, and
 assign a color for $b_1$ based on the
label of edge $(w_1, b_1)$. Now vertex $b_1$ has two white vertices $w_2$ and $w_3$ within distance $2$. Using Lemma~\ref{pp:p151}
we can (uniquely) extend the coloring of $b_1$ to $w_2$ (symmetrically, to $w_3$) so that additional black vertex $b_2$ gets
a color. Again, the coloring of $b_2$ can be extended to its nearest white neighbor. We continue such a propagation
of colors, see Figs.~\ref{fig:hex51b}~and~\ref{fig:hex51c} where processed black vertices and edges are shown dashed red.
One can easily see that the process will color a row of hexagons with alternate upper and lower legs.

To complete the coloring of $G$ we choose a white vertex in the next row of hexagons and initiate a similar propagation
process. For example, one can use vertices $w$ and $b$ shown in Fig.~\ref{fig:hex51c}.
\end{proof}

\subsection{Octagonal-Square Grid}

\begin{lemma}
\label{lm:84grid51}
Any octagonal-square grid is $(5,1)$-total-threshold-colorable.
\end{lemma}

\begin{figure}[htb]
    \center
	\subfigure[]{
	\includegraphics[height=3.5cm]{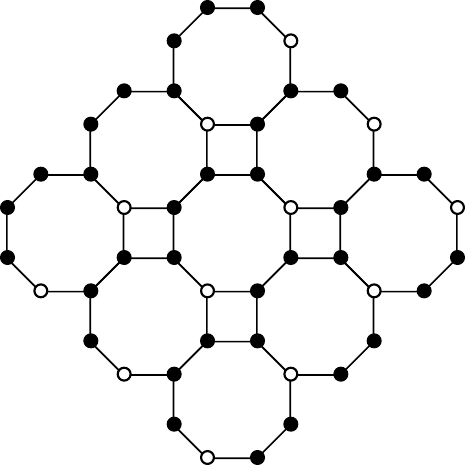}
    \label{fig:8451a}}
~
    \subfigure[]{
    \includegraphics[height=3.5cm]{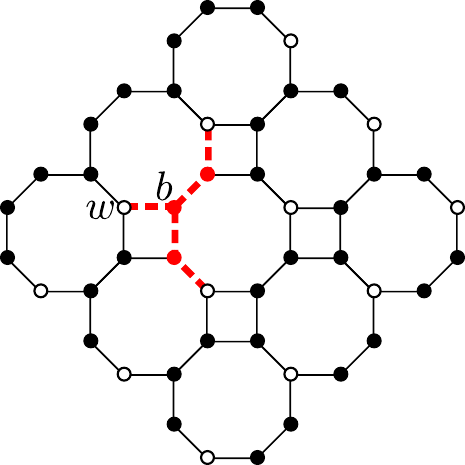}
    \label{fig:8451b}}
~
    \subfigure[]{
    \includegraphics[height=3.5cm]{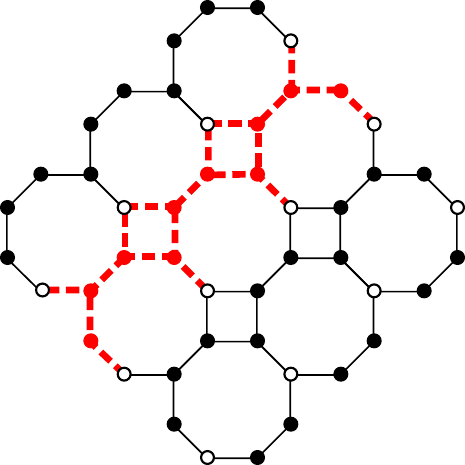}
    \label{fig:8451c}}
\caption{
Total-threshold-coloring of the octagonal-square grid.}
\end{figure}

\begin{proof}
We use colors $\{-2,-1,0,1,2\}$ and threshold $t=1$ to find a coloring; the proof is similar to the proof of Lemma~\ref{lm:hex51}.
We start by partitioning the vertices of the graph into white and black as shown in Fig.~\ref{fig:8451a},
and we assign color $0$ to the white ones. Then we choose a white vertex $w$ and its black neighbor
$b$ as in Fig.~\ref{fig:8451b}, and we assign colors $\{-2, -1, 1, 2\}$ to the ``row'' of black
vertices. It is easy to see that the coloring of rows can be done independently; see~Fig.~\ref{fig:8451c}.
\end{proof}

\subsection{Square-Triangle Grid}

We prove that the graph in Fig.~\ref{fig:43} is not total-threshold-colorable.
Assume to the contrary that $c$ is a $(r,t)$-threshold-coloring. Without loss of generality let
$c(v_0) < c(u_0)$. Since $(v_1,u_0)$ is a far edge and $(v_0,x), (u_0,x)$ are near we have
$c(v_0) < c(x) < c(u_0)$. Similar argument shows that $c(v_1) < c(v_0) < c(x) < c(u_0) < c(u_1)$.
Then if $x<y$, we have $c(v_1) + t < c(x)$ and $c(x) + t < c(y)$, which implies
 $c(v_1) + 2t < c(y)$. This makes it impossible to find a color for $v_2$ near to both $v_1$
 and $y$. Similarly if $x>y$ then it is impossible to color $u_2$.

Theorem~\ref{th:color} summarizes the results in this section.

\begin{theorem}
\label{th:color} Paths, cycles, trees, fans, the hexagonal grid and the octagonal-square grid are total-threshold-colorable.
The triangular grid and triangle-square grid are not total-threshold-colorable.
\end{theorem}

\section{Planar Graphs without Short Cycles}
\label{sec:WSC}

In the cases where we have counter-examples of total-threshold-colorability (e.g., $K_4$ and the triangular grid) we
have short cycles, which can be used to force groups of vertices to be simultaneously near and far. In this
section we show that if we consider graphs without short cycles, we can prove total-threshold-colorability.

\begin{theorem}
\label{thm:no9cyc}
Let $G$ be a planar graph without cycles of length $\le 9$.
Then $G$ is $(8,2)$-total-threshold-colorable\footnote{Equivalently, the \emph{girth} (that is, the shortest cycle) of $G$
should be $\ge$ 10.}.
\end{theorem}

The outline of our proof for Theorem~\ref{thm:no9cyc} is as follows. We first find some small
tree structures $T$ that are ``reducible'', in the sense that for any edge-labeling of $T$ and
any given fixed coloring of the leaves of $T$ to the colors $\{0, 1, \ldots, 7\}$, there is a
$(8,2)$-threshold-coloring
of $T$. For a contradiction assume that there is a planar graph with girth $\ge 10$ having
no $(8,2)$-threshold-coloring. We consider the minimal such graph $G$, and by a discharging
argument prove that $G$ contains at least one of these reducible tree structures. This
contradicts the minimality of $G$. We start with some technical claims.

\noindent
\paragraph{Extending a coloring.} Let $P_n$ be a path with vertices $v_0, \dots, v_n$.
Given an edge-labeling of $P_n$ and the color $c_0$ of $v_0$ we call a color $c_n$ \emph{legal} if there exists
a $(8,2)$-threshold-coloring $c$ of $P_n$, so that $c(v_0)=c_0$ and $c(v_n)=c_n$.

\begin{pp}
\label{pp:1}
Let $P_1$ be a path of length $1$. Then at least one of the colors $1$ or $6$ is legal (irrespective of the edge label and the color $c_0$).
\end{pp}

\begin{proof}
One only needs to observe that color $1$ is close to $0,1,2,3$, and is far from $4,5,6,7$, i.e. the distance between colors is at most 2
or strictly more than 2, respectively. The result follows by symmetry.
\end{proof}

\begin{pp}
\label{pp:2}
Let $P_2$ be a path of length $2$. Then $3$ is legal unless $c_0=3$ and $\{N,F\}=\{\{e_1\},\{e_2\}\}$, i.e.\ the
edges $e_1$ and $e_2$ are labeled differently. Symmetrically, $4$ is legal unless $c_0=4$ and
$\{N,F\}=\{\{e_1\},\{e_2\}\}$.
\end{pp}

\begin{proof}
By symmetry we only give the proof for the case $c_2=3$. If $N=\{e_1,e_2\}$ then we choose $c(v_1)$ to be the average of
$c_0$ and $c_2$, rounding if necessary. If $F=\{e_1,e_2\}$, then one of $0,7$ is a good choice for $c(v_1)$, as both 0,7 are far
from $c_2=3$, and at least one is far from $c_0$. In the remaining case we may assume that $c_0\ne 3$. If $c_0 < 3$, then set
$c(v_1)=0$ or $c(v_1)=5$ in case $e_2 \in F$ or $e_2 \in N$, respectively. If $c_0 > 3$, then set $c(v_1)=6$ or $c(v_1)=1$ in case $e_2 \in F$ or $e_2 \in N$, respectively.
\end{proof}

\begin{pp}
\label{pp:3}
Let $P_3$ be a path of length $3$. Then $1,3,4,$ and $6$ are all legal (irrespective of the edge label and the color $c_0$).
\end{pp}
\begin{proof}
By symmetry it is enough to find appropriate coloring extensions for which $c(v_3)=1$ and $c(v_3)=3$.
For the latter, choose $c_1=c(v_1) \ne 3$, according to $c_0$ and the label of $e_1$. Now by~\ref{pp:2} this choice of $c_1$ can be
extended to the remaining part of $P_3$, so that $c(v_3)=3$.
The goal $c(v_3)=1$ splits into two subcases. If $c_0 \ne 3,4$, then by~\ref{pp:2} both $3$ and $4$ are possible color choices for
$c(v_2)$. One is close and the other is far from $1$. In case $c_0$ is either $3$ or $4$, then again by~\ref{pp:2} both $1$ and $6$
are possible choices for $c(v_2)$. Again, the former is close and the latter is far from $1$.
\end{proof}

A \emph{star} is a subdivision of the graph $K_{1,n}$, and its \emph{center} is the single vertex of degree $\ge 3$.
Let $T$ be a star. A \emph{prong} of $T$ is a path
from a leaf to the center of $T$, and a prong
 with $k$ edges is called a $k$-prong, we say that it has length $k$.

\begin{pp}
\label{pp:3a}
Let $T$ be a subdivision of $K_{1,3}$ with prongs of length $1,2,$ and $3$, respectively. Assume that the leaves of
$T$ are assigned colors, so that the leaf $u$ on the $1$-prong is colored with either $1$ or $6$. Then we can extend
this partial coloring to the whole $T$.
\end{pp}
\begin{proof}
Let $v$ be the center of $T$. Given $c(u)$, we can choose $c(v) \in \{3,4\}$ so that the labeling condition on
the $1$-prong is satisfied. If this choice cannot be extended to the longer prongs, then the leaf of the 2-prong
is also colored with either $3$ or $4$, see~\ref{pp:2}. But then the choice $c(v) \in \{1,6\}$ which satisfies
the labeling condition on the $1$-prong can be extended to the remaining prongs.
\end{proof}


\noindent
\paragraph{Reducible configurations.} A \emph{configuration} is a tree $T$, and is \emph{reducible} if every assignment of colors to the leaves of $T$
can be, for every possible edge-labeling of $T$, extended to a $(8,2)$-threshold-coloring $c$ of the whole $T$.

\begin{pp}
\label{pp:4}
A path $P_4$ of length 4 is a reducible configuration.
\end{pp}
\begin{proof}
Let $v$ be a neighbor of a leaf in $P_4$. By~\ref{pp:1} and \ref{pp:3} either $c(v)=1$ or $c(v)=6$ extends to the remaining uncolored vertices.
\end{proof}

Now~\ref{pp:4} implies that longer paths are reducible as well. Let us turn our attention to stars.
%
%

\begin{pp}
\label{pp:5}
\begin{enumerate}
\item[(A)] Let $T$ be a star with at most 1 prong of length $1$ and the remaining prongs have
 length $3$. Then $T$ is reducible.
\item[(B)] Let $T$ be a star with at most 3 prongs of length $2$ and the remaining prongs have length $3$. Then $T$ is reducible.
\end{enumerate}
\end{pp}
\begin{proof}
In both cases let $v$ denote the center of the star.
In order to establish (A) let $c(v)$ be either $1$ or $6$, which is appropriate for the 1-prong (such a choice exists
by~\ref{pp:1}). By~\ref{pp:3} the coloring $c(v)$ can be extended to the remaining 3-prongs.
For (B) we may assume that neither $3$ nor $4$ can be extended to all three $2$-prongs. By~\ref{pp:2} both colors $3$
and $4$ are used at leaves of the $2$-prongs. Now, by~\ref{pp:1} at least one of $c(v)=1$ or $c(v)=6$ extends to the
third 2-prong, and hence also to the remaining 2- and 3-prongs, by~\ref{pp:2} and~\ref{pp:3}.
\end{proof}

\begin{pp}
\label{pp:5a}
There exist two additional types T1 and T2 of reducible configurations shown in Fig.~\ref{fig3}.
\end{pp}
\begin{figure}[t]
    \center
	\subfigure{
    \includegraphics[height=2.3cm]{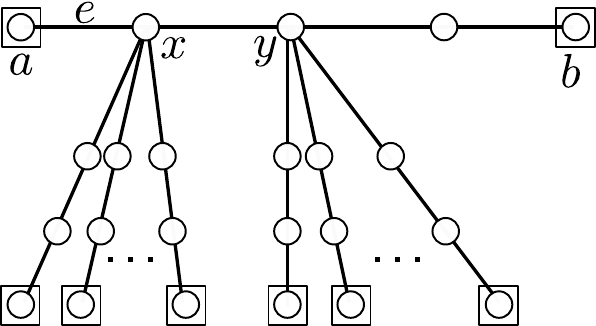}}
~~~~~~~~~~~~~~~~~
    \subfigure{
    \includegraphics[height=2.3cm]{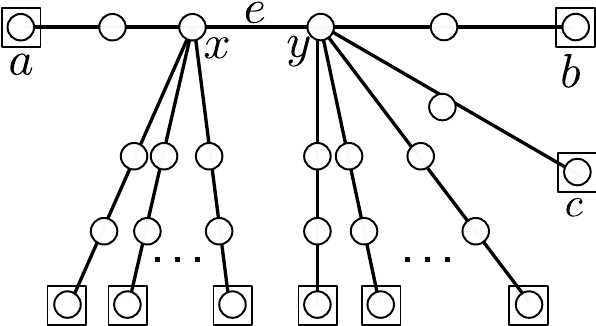}}
    \label{fig3}
    \caption{Two additional types of reducible configurations, T1 and T2.}
\end{figure}

\begin{proof}
Let us first consider the T1 configuration. By~\ref{pp:1} one of 1,6 is appropriate for the color of $x$, with
respect to color $a$ and type of edge $e$. If $b \in \{3,4\}$, then choose $c(y)$ from $1,6$, and if
$b \not\in \{3,4\}$ then choose $c(y)$ from $\{3,4\}$. By~\ref{pp:2} this works.

Let us now turn to T2. If $e$ is a near edge, we might as well contract $e$ (which implies both $x$ and $y$ will
receive the same color), and reduce to~\ref{pp:5}(B).

Hence we shall assume $e$ is a far edge. Assume first that coloring vertex $x$ with both $1$ and $6$ extends to
the left 2-prong at $x$. If $c(x)=1$ and $c(y)=4$ does not extend to the right $2$-prongs at $y$, we may assume
$b=4$. If $c(x)=6$ and $c(y)=3$ does not extend to the right $2$-prongs at $y$, we may assume $c=3$. In this
case setting $c(x)=1$ and $c(y)=6$ extends to the right.

By~\ref{pp:1} we may assume that only one of $c(x)=1$ or $c(x)=6$ extends to the left 2-prong at $x$, without loss of
generality the former. Now $a \ne 3$ and $a \ne 4$, and both $c(x)=3$ and $c(x)=4$ extend left.
A choice of $c(x)=1, c(y)=4$ does not extend to the right $2$-prongs at $y$ only if, say, $b$ is equal to $4$. But now
at least one of $c(y)=1$ or $c(y)=6$ extends to the right 2-prongs, and such a choice can be complemented with $c(x)=4$
or $c(x)=3$, respectively.
\end{proof}

\noindent
\paragraph{Discharging.} A \emph{minimal counterexample} is the smallest possible (in terms of order) planar graph $G$ without cycles of length
$\le 9$ which is not $(8,2)$-total-threshold-colorable. A minimal counterexample $G$ does not contain reducible configurations. Further $G$ is connected and has no vertices of degree $1$. As $G$ is also not a
cycle (such a cycle should be of length $\ge 9$ and should not contain a $P_4$), and is therefore homeomorphic to a
(multi)graph of minimal degree $\ge 3$.

Let us fix its planar embedding determining its set of faces $F(G)$.
Let us define \emph{initial charges}: initial charge of a vertex $v$, $\charge_0(v)$, is equal to $4\deg(v)-10$, and the
initial charge of a face $f$, $\charge_0(f)$, is equal to $\deg(f) -10$. A routine application of Euler formula shows
that the total initial charge is $-20$.

As all faces have length $\ge 10$, every face is initially non-negatively charged. We shall not alter the charges of faces.

The following table shows initial charges of vertices according to their degree:
\begin{center}
\begin{tabular}{l|rrrrrrr}
degree $\deg(v)$ & 2 & 3 & 4 & 5 & 6 & 7 & $\cdots$
\\ \hline
initial charge $\charge_0(v)$ & $-2$ & $2$ & $6$ &
$10$ & $14$ & $18$ & $\cdots$
\end{tabular}
\end{center}

The discharging procedure will run in two phases, by $\charge_i(v)$ we shall denote the charge of vertex $v$ after Phase $i$ of
discharging. Informally, Phase 1 shall see that vertices of degree 2 do not have negative charges, and Phase 2 will
leave only vertices of degree 3 with a possible negative charge.

Let $u, v$ be vertices of $G$. We say that $u$ and $v$ are $2$-adjacent, if $G$ contains a $u-v$-path  whose (possible)
internal vertices all have degree $2$. In Phase 1 we redistribute charge according to the following rule:
\medskip

\noindent
{\bf Rule 1}: every vertex $v$ of degree $\ge 3$ sends charge $1$ to every vertex $u$ of degree $2$, for which $v$ and $u$ are \emph{$2$-adjacent}.
\medskip

In {Phase 2} we shall apply the following rule:
\medskip

\noindent
{\bf Rule 2}: If $u$ and $v$ are adjacent with $\charge_1(u) > 0, \charge_1(v) < 0$ then $u$
 sends charge $1$ to $v$.
\medskip

As every vertex $u$ of degree 2 (we also call them \textit{2-vertices}) is 2-adjacent to exactly two vertices of bigger degree, we have $\charge_1(u) = 0$ in this
case. For a vertex $v$ of degree $\ge 3$, the discharging in Phase 1 decreases the charge of $v$ by the number of
2-vertices which are 2-adjacent to $v$.

Let $v$ be a vertex of degree $\ge 3$. A \emph{prong at $v$} is a $v-x$-path whose other end-vertex $x$ is of
degree $\ge 3$ and has internal vertices of degree 2.
\begin{pp}
\label{pp:6}
Let $v$ be a vertex of degree $\ge 3$. Then the number of 2-vertices that are 2-adjacent to $v$ is at most $2 \cdot \deg(v) - 3$.
\end{pp}
\begin{proof}
By~\ref{pp:4} each prong at $v$ contains at most two vertices of degree $2$. If the shortest prong at $v$ has length $1$,
then~\ref{pp:5} implies that at least one other prong has length $\le 2$. If the shortest prong at $v$ has length $2$,
then by~\ref{pp:5} we have at least four prongs that are of length $\le 2$, and the result follows.
\end{proof}

Now~\ref{pp:6} serves as the lower bound for vertex charges after Phase 1, and in turn prepares us for the Phase 2 of discharging. \begin{pp}
\label{pp:7}
\begin{enumerate}
\item[(A)] Let $v$ be a vertex of degree $3$. If $\charge_1(v) < 0$, then $\charge_1(v) = -1$ and the prongs at $v$ have lengths $1,2$ and $3$, respectively.
\item[(B)] Let $v$ be a vertex of degree $3$. If $\charge_1(v) = 0$, then the prongs at $v$ have either lengths $1,1,3$ or $1,2,2$.
\item[(C)] Let $v$ be a vertex of degree $3$ with its prongs of length $1,1,$ and $2$. Then $\charge_1(v) = 1$.
\item[(D)] Let $v$ be a vertex of degree $3$ with all $3$ prongs of length $1$. Then $\charge_1(v) = 2$.
\item[(E)] If $v$ is a vertex of degree $\ge 4$, then $\charge_2(v) \ge 0$, and also $\charge_2(v)$ is not smaller than the number of $1$-prongs at $v$.
\end{enumerate}
\end{pp}
\begin{proof}
Let us first prove (E). Choose a vertex $v$ with $\deg(v) \ge 4$. For every prong of length $3$, $v$ sends 2 units of
charge in Phase $1$. For every shorter prong $v$ sends at most $1$ unit of charge in either Phase 1 or Phase 2. The
total charge sent out of $v$ in both of the phases is by~\ref{pp:5} and~\ref{pp:6} at most
 $2 \deg(v) - 2$. Hence
$\charge_2(v) \ge (4 \deg(v) - 10)-(2 \deg(v)-2)=2 \deg(v)-8 \ge 0$.

The other cases merely stratify vertices of degree $3$ according to the number of their $2$-neighbors of degree $2$.
\end{proof}

Now~\ref{pp:7}(E) states that every vertex $v$ of degree $\ge 4$ satisfies $\charge_2(v) \ge 0$. Similarly, if a $3$-vertex
$u$ is adjacent to a vertex $v$ whose degree is at least $4$, then also $\charge_2(u) \ge 0$. This fact follows from
either~\ref{pp:7}(A) and (E) (in case $\charge_1(u)<0$), or from either~\ref{pp:7}(C) or (D) (if $\charge_1(v)>0$) as in this case
$u$ cannot send excessive charge in Phase 2.

\begin{pp}
\label{pp:8}
No vertex $v$ has $\charge_2(v)<0$ and $\charge_1(v)<0$.
\end{pp}
\begin{figure}[t]
    \center
	\subfigure{
    \includegraphics[height=2.2cm]{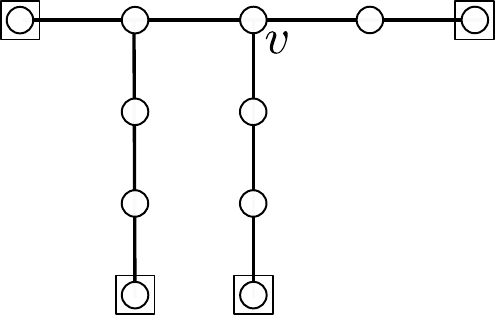}}
~~~~~
    \subfigure{
    \includegraphics[height=2.2cm]{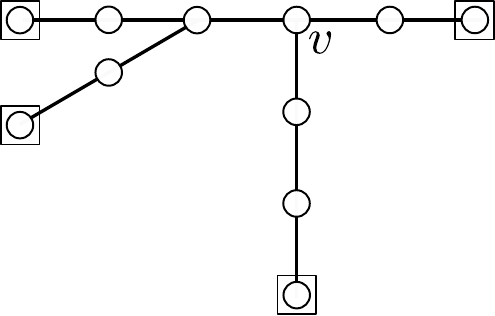}}
~~~~~
    \subfigure{
    \includegraphics[height=2.2cm]{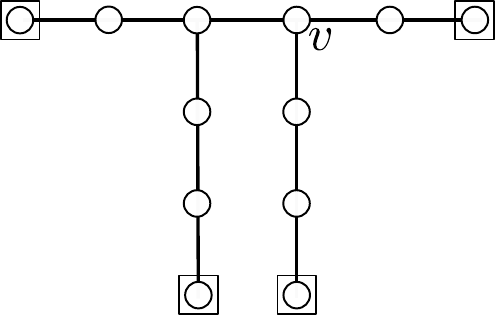}}
    \label{fig4}
    \caption{Negatively charged vertex $v$ after both phases induces a reducible configuration.}
\end{figure}
\begin{proof}
Let $v$ be a vertex satisfying both $\charge_2(v)<0$ and $\charge_1(v)<0$. By~\ref{pp:7} $\deg(v) =3$ and
 $v$ has prongs of length $1,2,3$. Let $u$ be the only neighbor of $v$ of degree $\ne 2$. Since $v$ has received no charge from $u$ in Phase 2 we have both $\deg(u)=3$ and $\charge_1(u)\le 0$. By~\ref{pp:7} the prongs of $u$ are of lengths $1,2,3$
or $1,1,3$ or $1,2,2$. Hence $G$ contains one of configurations shown in Fig.~\ref{fig4}.

Now observe that these are reducible, as each matches one of T1 or T2 types of reducible configurations~\ref{pp:5a}.
\end{proof}

\begin{pp}
\label{pp:9}
No vertex $v$ has $\charge_2(v)<0$ and $\charge_1(v)\ge 0$.
\end{pp}
\begin{proof}
If $\charge_1(v)=0$, then also $\charge_2(v)=0$, as Rule 2 does not reduce charge of a discharged vertex. By~\ref{pp:7}(E) vertices of
degree $\ge 4$ do not have negative charge after Phase 2.

Hence we may assume that $v$ has degree $3$, $\charge_1(v)> 0$, and $\charge_2(v)<0$.
By~\ref{pp:7}(C) and (D) every neighbor $u$ of $v$ satisfies either $\deg(u)=2$ or $\deg(u)=3$ and $\charge_1(u)<0$.
There are exactly two possible cases and they are shown in Fig.~\ref{fig5}.
\begin{figure}[t]
    \center
	\subfigure{
    \includegraphics[height=2.9cm]{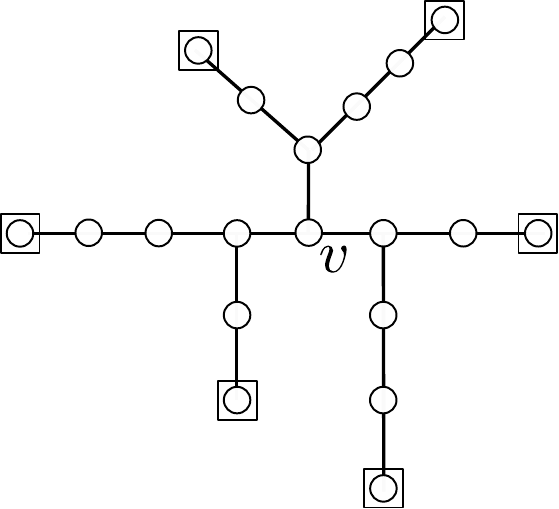}}
~~~~~~~~~~~~~~~~~~~~~~
    \subfigure{
    \includegraphics[height=2.9cm]{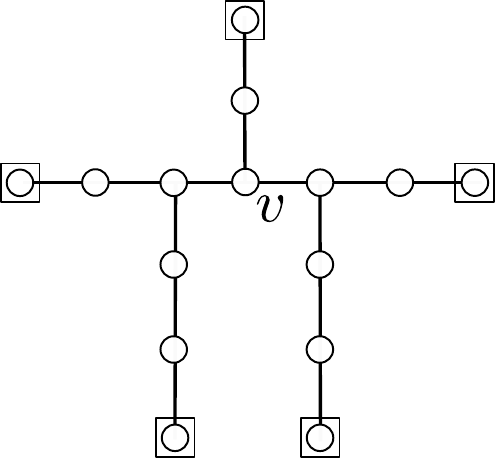}}
    \label{fig5}
    \caption{Negatively charged vertex $v$ after Phase 2, its charge was positive after Phase 1.}
\end{figure}
It is enough to see that there exists a color choice $c(v)$ which can be extended in the 2-prong and/or
stars centered at neighbors of $v$.

Let us first settle the option shown in the right. By~\ref{pp:1} at least one of $c(v)=1$ or $c(v)=6$ extends to
the top $2$-prong, and this choice also extends to the two copies of $T$, see~\ref{pp:3a}.
The left case is even easier, as both choices $c(v)=1$ and $c(v)=6$ extend to the three copies of $T$, again by~\ref{pp:3a}.
\end{proof}

Now~\ref{pp:8} and \ref{pp:9} imply that no vertex has negative charge after Phase 2 of the discharging procedure.
As the total charge remains negative and the faces cannot have negative charges we have a contradiction, which completes
the proof of Theorem~\ref{thm:no9cyc}.

\section{Unit-Cube Contact Representations of Graphs}
\label{sect:unit}

\begin{lemma}
\label{lem:threshold-cube}
If $G$ has a unit-cube contact representation $\Gamma$
 so that one face of each cube is co-planar in $\Gamma$, then any threshold subgraph of
 $G$ also has a unit-cube representation.
\end{lemma}
\begin{proof} Let $H=(V,E_H)$ be a threshold subgraph of $G=(V,E_G)$ and let
 $c:V\rightarrow [1\dots r]$ be an $(r,t)$-threshold-coloring of $G$ with respect to the edge-partition
 $\{E_H, E_G-E_H\}$. We now compute a unit-cube contact representation of $H$ from
 $\Gamma$ using $c$.

Assume (after possible rotation and translation) that the bottom face for each
 cube in $\Gamma$ is co-planar with the plane $z=0$; see Fig.~\ref{fig:threshold-cube}(a).
 Also assume (after possible scaling)
 that each cube in $\Gamma$ has side length $t+\epsilon$, where $0<\epsilon<1$.
 Then we can obtain a unit-cube contact representation of $H$ from $\Gamma$ by lifting
 the cube for each vertex $v$ by an amount $c(v)$ so that its bottom face is at $z=c(v)$; see
 Fig.~\ref{fig:threshold-cube}(b).
 Note that for any edge $(u,v)\in E_H$, the relative distance between the bottom faces of the
 cubes for $u$ and $v$ is $|c(u)-c(v)|\le t<(t+\epsilon)$; thus the two cubes maintain contact. On the other hand, for each pair of vertices $u,v$ with
 $(u,v)\notin E_H$, one of the following two cases occurs: (i) either $(u,v)\notin E_G$ and
 their corresponding cubes remain non-adjacent as they were in $\Gamma$; or
 (ii) $(u,v)\in (E_G-E_H)$ and the relative distance between the bottom faces of the two
 cubes is $|c(u)-c(v)|\ge (t+1)>(t+\epsilon)$, making them non-adjacent.
\end{proof}

\begin{corollary}
\label{cor:cube-thr} Any subgraph of hexagonal and octagonal-square grid
 has a unit-cube contact representation.
\end{corollary}
\begin{proof} This follows from the fact that each of these grids has a unit-cube contact
 representation [Fig.~\ref{fig:grid-cube}(c)--(d)] and is total-threshold-colorable
 [Lemmas~\ref{lm:hex51},~\ref{lm:84grid51}].
\end{proof}

\begin{figure}[t]
\centering
\includegraphics[width=\textwidth]{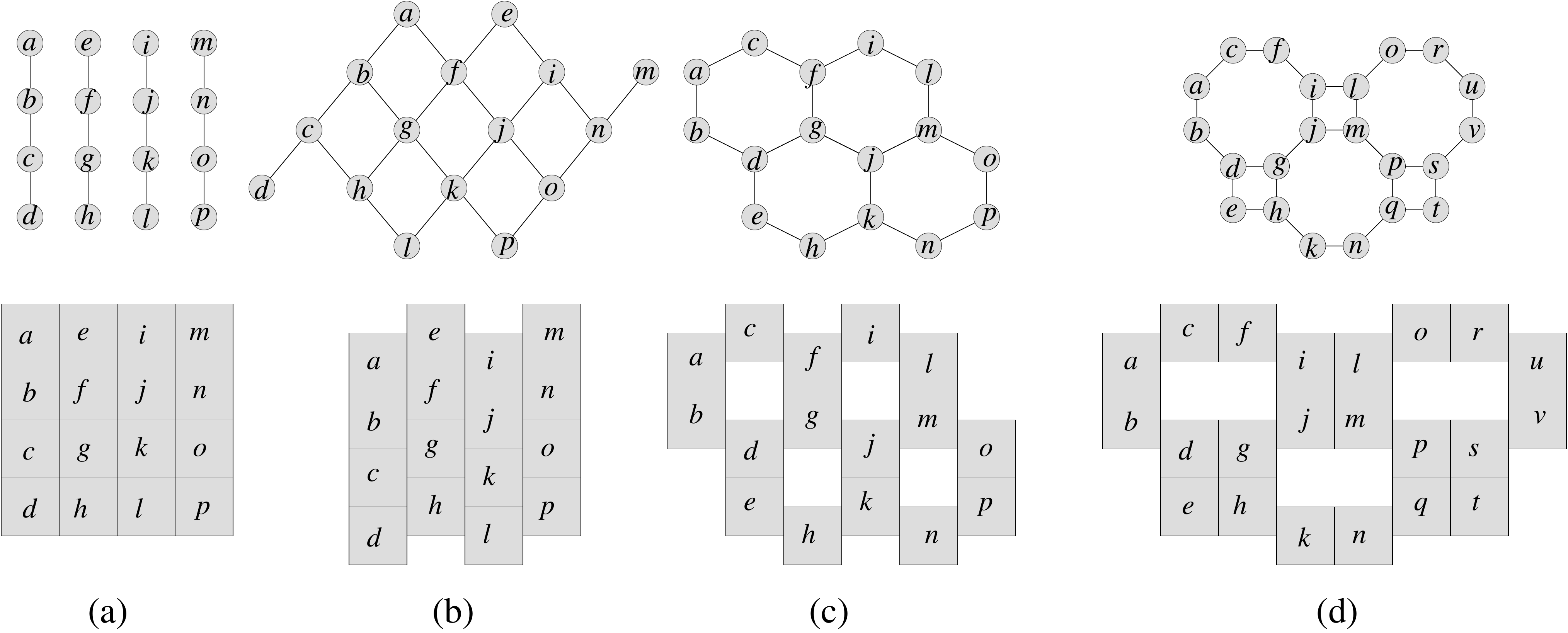}
\caption{Unit-cube contact representations for square, triangular, hexagonal and
 octagonal-square grids (top-view).}
\label{fig:grid-cube}
\end{figure}

Unfortunately, triangular grids are not always threshold-colorable, while for square grids,
 the status of the threshold-colorability is unknown. Thus we cannot use the result of
 Lemma~\ref{lem:threshold-cube} to find unit-cube contact representations for the subgraphs
 of square and triangular grids, although there are nice unit-cube contact representations
 for these grids with co-planar faces; see~\ref{fig:grid-cube}(a)--(b). Instead of using the threshold-coloring
 approach, we next show how to directly compute a unit-cube contact representation
via geometric algorithms.
 Specifically, we describe such geometric algorithms for unit-cube contact
 representations for any subgraph of hexagonal and square grids.

\begin{lemma}
\label{lem:unit-hexa}
 Any subgraph of a hexagonal grid has a unit-cube contact representation.
\end{lemma}
\begin{proof}
This claim follows from Corollary~\ref{cor:cube-thr}. Furthermore since hexagonal grids
 are subgraphs of square grids, this result can also be proven as a corollary of
 Lemma~\ref{lem:unit-square}. However here we give an alternative proof by designing
 a geometric algorithm to construct unit-cube contact representation for subgraphs
 of hexagonal grids.

\begin{figure}[t]
\centering
\includegraphics[height=5cm]{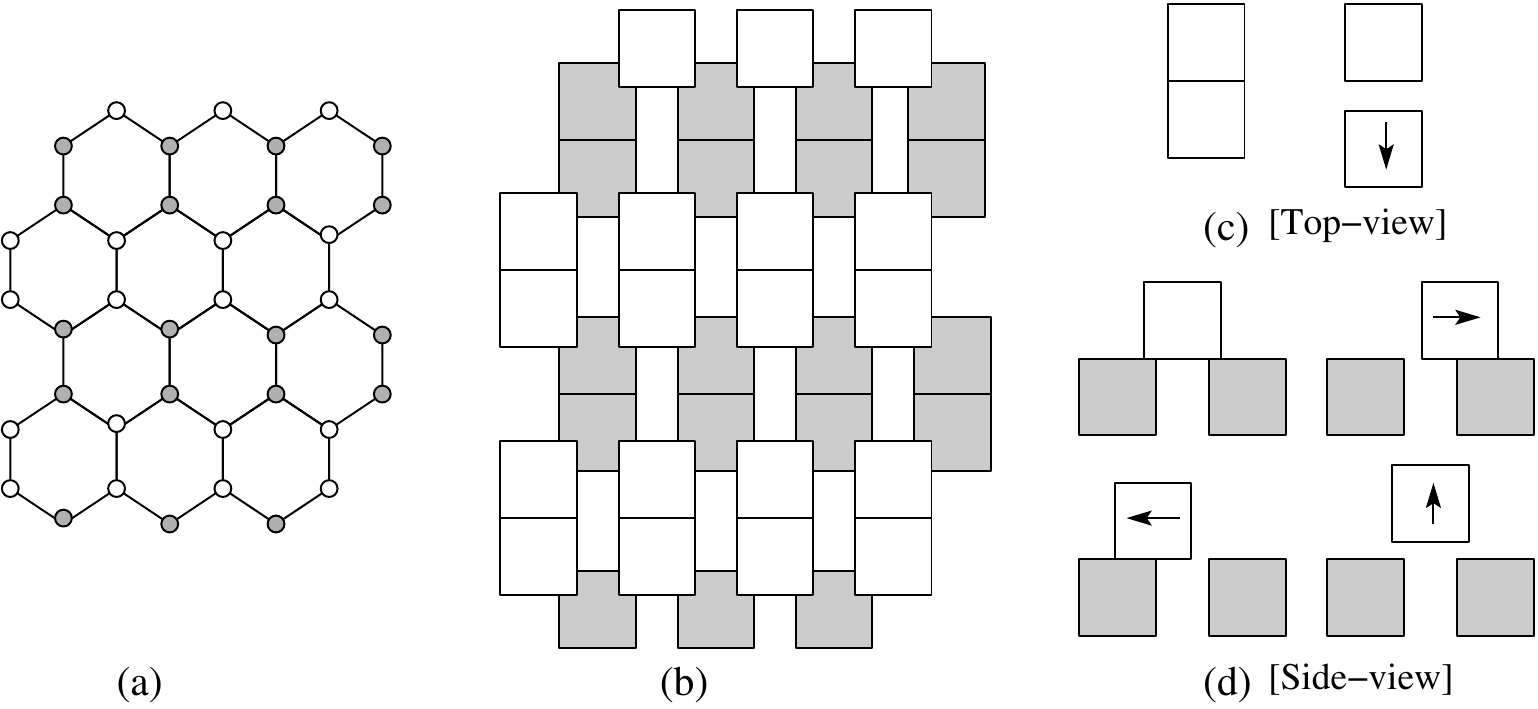}
\caption{Construction of unit-cube representation for any subgraph of a hexagonal grid.}
\label{fig:unit-hexa}
\end{figure}

Let $G$ be a subgraph of a hexagonal grid $G^*$, as in Fig.~\ref{fig:unit-hexa}(a).
 We first construct a unit-cube representation of $G^*$, where the base of each gray
 cube has $z$-coordinate 0 and the base of each white cube has $z$-coordinate 1;
 see Fig.~\ref{fig:unit-hexa}(b). Call this representation $\Gamma^*$. We now obtain
 a representation of $G$ from $\Gamma^*$ as follows. First, we delete the cubes
 corresponding to the vertices of $G^*$ that are not in $G$. Now to delete the edges
 in $G^*$ not in $G$, we note that $G^*$ is bipartite. Let $A$ be a partite set of $G$.
 Then we can delete a set of edges by only removing the contact from cubes corresponding
 to vertices in $A$. Suppose $v$ is a vertex in $A$ and $R$ is the corresponding cube.
 Without loss of generality, assume that $R$ is a white cube. Then $R$ has (at most) three
 contacts: one with a white cube $w$ and two other with two gray cubes $g_1$, $g_2$.
 To get rid of the contact with $w$, we just shift $R$ a small distance $0<\epsilon<0.5$
 away from $w$; see Fig.~\ref{fig:unit-hexa}(c). On the other hand, to get rid of the
 contact with exactly one of $g_1$ and $g_2$, we shift $R$ away from that cube until
 it looses the contact, while to get rid of both the adjacencies, we shift $R$ a small distance
 $0<\epsilon<0.5$ upward; see Fig.~\ref{fig:unit-hexa}(d).
\end{proof}

\begin{lemma}
\label{lem:unit-square}
 Any subgraph of a square grid has a unit-cube contact representation.
\end{lemma}
\begin{proof}
 Let $G$ be a subgraph of a square grid $G^*$. We first construct
 a unit-cube representation of $G^*$. Note that $G^*$ is a bipartite graph and
 suppose $A$ and $B$ are its two partite sets; see Fig.~\ref{fig:unit-square}(a).
 We first place cubes for the vertices of the set $A$ (white and gray vertices in the figure).
 Consider the $(i,j)$ coordinate system on these vertices, as illustrated in
 Fig.~\ref{fig:unit-square}(a), with the center $(0,0)$ taken arbitrarily.
 Consider a vertex $v$ of $A$ with the coordinate $(i,j)$ in this coordinate system.
 We call $v$ a white vertex when $i+j$ is even; otherwise $v$ is a gray vertex.
 We place the cube for $v$ in the range
 $[x(i,j), x(i,j)+1]\times [y(i,j), y(i,j)+1]\times [z(i,j), z(i,j)+1]$;
 where $x(i,j)=[\delta i + (2-\delta) j]$,
 $y(i,j)=[(2-\delta) i -\delta j]$, $0<\delta<0.5$ and $z(i,j)=0.5$ if $v$ is a white vertex,
 otherwise $z(i,j)=0$.
 Each black vertex has two adjacent white and two adjacent gray vertices and we
 place its unit cube between the cubes for four adjacent vertices with $z$-coordinate 0.25;
 see Fig.~\ref{fig:unit-square}(b).

\begin{figure}[t]
\centering
\includegraphics[width=0.6\textwidth]{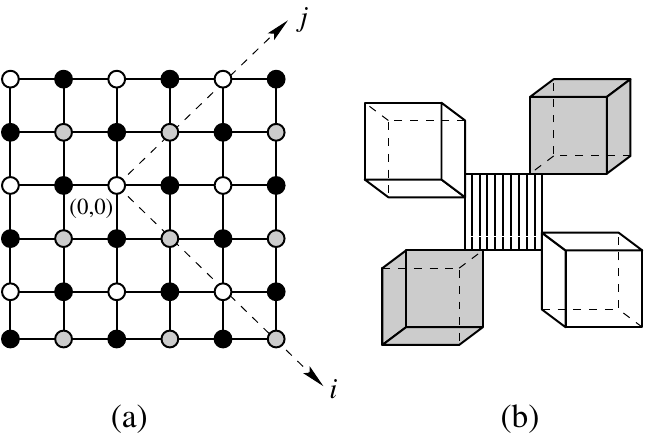}
\caption{Construction of unit-cube representation for any subgraph of a square grid.}
\label{fig:unit-square}
\end{figure}

We now modify this representation of $G^*$ to compute a representation for $G$.
 First we delete the cubes for the vertices of $G^*$ not in $G$. Then to remove contacts
 corresponding to the edges of $G^*$ not in $G$, we move cubes for the black vertices.
 Suppose $v_b$ is such a black vertex and its two adjacent white vertices are at coordinates
 $(i,j)$, $(i+1,j+1)$; while its two adjacent gray vertices are at coordinates $(i,j+1)$, $(i+1,j)$.
 Call the cube for the black vertex $R_b$ and let $R(x,y)$ denote the cubes for the white
 or gray vertex at coordinate $(x,y)$. If $v_b$ has degree $4$ in $G$, we don't have to move
 $R_b$. Again if $v_b$ has degree 0 in $G$, we move $R_b$ outside the boundary of
 $\Gamma$. Otherwise depending on the incident edges of $v_b$ present in $G$, we need to
 get rid of some of the contacts with $R_b$. We show how to do this in cases.

\noindent
\textbf{Case 1: $v_b$ has degree 3 in $G$.}
 Assume that the incident edge of $v_b$ in $G^*$ missing in $G$ is with $v(i,j)$.
 Then we shift $R_b$ downward until its topmost plane has $z$-coordinate 0.5
 and then we shift $R_b$ in the $xy$-plane by a small amount $\epsilon$ away
 from $R(i,j)$, where $0<\epsilon<\delta$.

\noindent
\textbf{Case 2: $v_b$ has degree 2 in $G$.} If the two incident edges of $v_b$ missing in $G$
 are both with white (gray) vertices, then we shift $R_b$ downwards (upwards) until it looses
 contacts with both its white (gray) neighbors. Otherwise, assume that the two incident edges
 of $v_b$ missing in $G$ are with $v(i,j)$ and $v(i,j+1)$. In this case, we shift $R_b$ downward
 until its topmost plane has $z$-coordinate 0.5 and then we shift $R_b$ in the $xy$-plane
 away from both $R(i,j)$ and $R(i,j+1)$ until it looses contact with both of them.

\noindent
\textbf{Case 3: $v_b$ has degree 1 in $G$.} Suppose the incident edges of $v_b$ missing
 from $G$ are with $v(i,j)$, $v(i,j+1)$ and $v(i+1,j)$. We then first shift $R_b$ upward
 until its bottommost plane has $z$-coordinate is in the open interval $(1, 1.5)$ so that
 it looses contacts with both $R(i,j+1)$ and $R(i+1,j)$. Finally we move $R_b$ away from
 $R(i,j)$ until it looses the contact with $R(i,j)$.

\end{proof}

To summarize the results in this section:

\begin{theorem}
\label{th:unit-cube} Any subgraph of the square, hexagonal and octagonal-square grid
 has a unit-cube contact representation.
\end{theorem}

\section{Conclusion and Open Problems}

We introduced a new graph coloring problem, called threshold-coloring, that generates
spanning subgraphs from an input graph where the edges of the subgraph are implied
by small absolute value difference between the colors of the endpoints.
We showed that any spanning subgraph of trees, some planar grids, and planar graphs
without cycles of length $\le 9$
can be generated in this way; for other classes like triangular and square-triangle
grids, we showed that this is not possible.
We also considered different variants of the problem and noted relations with other well-known graph coloring and
graph-theoretic problems.
Finally we use the threshold-coloring problem to find unit-cube contact representation for all the subgraphs
of some planar grids. The following is a list of some interesting open problems and future work.

\begin{enumerate}
\item Some classes of graphs are total-threshold-colorable, while others are not. There are many classes for which the problem
remains open; see Table~\ref{tbl:results} for some examples. A particularly interesting
 class is the square grid: does any subgraph of a square grid have a threshold-coloring?

\item Theorem~\ref{thm:no9cyc} implies that any planar graph without cycles of length $\le 9$
 is total-threshold-colorable. On the other hand, all our examples of non-threshold-colorable
 in Fig.~\ref{fig:trian} contain triangles. Can we reduce this
 gap by identifying the minimum cycle-length (girth) in a planar graph that guarantees total-threshold-colorability?

\item Can we efficiently recognize graphs that are threshold-colorable?

\item Is there a good characterization of threshold-colorable graphs?

\item The triangular and square-triangular grid are not total-threshold-colorable and we cannot use threshold-colorability
to find unit-cube contact representations; can we give a geometric algorithm (such as those in Section~\ref{sect:unit}
for square and hexagonal grids) to directly compute such representations?

\end{enumerate}
\bigskip

\noindent{\bf\large Acknowledgments:}
We thank Torsten Ueckerdt, Carola Winzen and Michael Bekos for discussions about different variants of
the threshold-coloring problem.

\bibliographystyle{abbrv}

\end{document}